\documentclass[11pt]{article}
\usepackage[left=1in, top=1in, right=1in, bottom=1in]{geometry}
\usepackage{indentfirst}
\usepackage{amsmath, amssymb, amsthm, amsfonts}
\usepackage[numbers,sort&compress]{natbib}
\usepackage{algpseudocode, algorithm}

\usepackage{tcolorbox, qtree}
\usepackage{enumerate}
\usepackage{fancyvrb}
\usepackage{caption}
\usepackage{xfrac}
\usepackage{bbm}

\usepackage[T1]{fontenc}

\makeatletter
\newtheorem*{rep@theorem}{\rep@title}
\newcommand{\newreptheorem}[2]{%
\newenvironment{rep#1}[1]{%
 \def\rep@title{#2 \ref{##1}}%
 \begin{rep@theorem}}%
 {\end{rep@theorem}}}
\makeatother

\newreptheorem{theorem}{Theorem}

\DeclareMathOperator*{\E}{\mathbf{E}}
\let\Pr\relax
\DeclareMathOperator*{\Pr}{\mathbf{Pr}}
\newtheorem{theorem}{Theorem}
\newtheorem{corollary}{Corollary}

\newtheorem{define}{Definition}
\newtheorem{lemma}{Lemma}

\newtheorem{claim}{Claim}

\newcommand{\N}{\mathcal{N}}
\newcommand{\C}{\mathcal{C}}

\title{Approximately Stable Committee Selection}
\date{}

\newcommand*\samethanks[1][\value{footnote}]{\footnotemark[#1]}

\author{Zhihao Jiang\thanks{Institute for Interdisciplinary Information Sciences, Tsinghua University. Email: \texttt{jzh16@mails.tsinghua.edu.cn}.} \and Kamesh Munagala\thanks{Department of Computer Science, Duke University. Email: \texttt{\{kamesh,knwang\}@cs.duke.edu}.} \and Kangning Wang\samethanks[2]}

\begin{document}

\thispagestyle{empty}
\maketitle

\pagenumbering{gobble}
\begin{abstract}
In the committee selection problem, we are given $m$ candidates, and $n$ voters. Candidates can have different weights. A committee is a subset of candidates, and its weight is the sum of weights of its candidates. Each voter expresses an ordinal ranking over all possible committees. The only assumption we make on preferences is {\em monotonicity}: If $S \subseteq S'$ are two committees, then any voter weakly prefers $S'$ to $S$.

We study  a general notion of group fairness via stability: A committee of given total weight $K$ is {\em stable} if no coalition of voters can deviate and choose a committee of proportional weight, so that all these voters strictly prefer the new committee to the existing one. Extending this notion to approximation, for parameter $c \ge 1$, a committee $S$ of weight $K$ is said to be $c$-approximately stable if for any other committee $S'$ of weight $K'$, the fraction of voters that strictly prefer $S'$ to $S$ is strictly less than $\frac{c K'}{K}$. When $c = 1$, this condition is equivalent to classical {\em core stability}.

The question we ask is: Does a $c$-approximately stable committee of weight at most any given value $K$ always exist for constant $c$? It is relatively easy to show that there exist monotone preferences for which $c \ge 2$. However, even for simple and widely studied preference structures, a non-trivial upper bound on $c$ has been elusive.

In this paper, we show that $c = O(1)$ for {\em all} monotone preference structures. Our proof proceeds via showing an existence result for a randomized notion of stability, and iteratively rounding the resulting fractional solution.
\end{abstract}

\newpage
\pagenumbering{arabic}
\setcounter{page}{1}

\section{Introduction}
Fair allocation of resources is a widely studied problem in social choice literature. In several societal decision-making scenarios, the resources are {\em public goods}, meaning that they can be enjoyed by multiple agents simultaneously. For instance, consider the problem of locating a fixed number of parks or libraries to serve a population~\cite{ChenFLM19}. Each such resource provides shared utility to several members of society. Viewed as a facility location problem, a standard objective involves locating the facilities to minimize the total distance traveled by the population to its nearest open facility. However, such a solution need not be {\em fair}: In a city with a dense urban core and sprawling suburbs, it can lead to the algorithm placing many more facilities in the suburbs, causing the locations at the urban core to become overcrowded. In other words, the globally optimal solution may produce disparate outcomes for different demographic slices.

Similarly, consider the {\em participatory budgeting} problem~\cite{PBP,FairKnapsack,implicitPB,knapsack1,knapsack2,Fain2016}. Recently, many cities and wards across the world have a process to put part of their budget to vote. The city chooses several projects such as repaving streets, installing lights, \emph{etc}, and each voter indicates preferences over these projects. The goal of the city is to choose a subset of projects that is feasible within the budget to fund. Again, simple schemes to aggregate voter preferences may overly bias the outcome towards majority preferences, and may ignore entirely the preferences of a sizable, coherent, minority.

\subsection{Committee Selection and Fairness Model}
In this paper, we consider an abstract resource allocation model -- {\em committee selection} -- that captures not only the above two settings, but also several other problems studied in social choice and in network design. We use the term ``committee selection'' based on similar terminology in social choice literature; however, as discussed below, our model captures general combinatorial selection problems.  In these settings, we study a classical notion of group fairness, and show that solutions that approximately satisfy this notion always exist.

\paragraph{Committee Selection.}  Using social choice parlance, the committee selection problem models the common scenario of determining a winning subset, \emph{i.e.} a \emph{committee}, from a set of candidates. A set of voters (or agents) $\N = [n] = \{1, \ldots, n\}$ and a set of candidates $\C = [m]$ are given, and each candidate $i$ is associated with a weight $s_i \ge 0$. Let $w(S) = \sum_{i \in S} s_i$ denote the weight of the committee $S \subseteq 2^{\C}$. The goal is to find a committee $S$ of weight at most a given value $K$. We note that our results extend to the setting where the weight $w(S)$ is a subadditive function of $S$; we present the additive model for simplicity of exposition.

Each voter $v \in \N$ explicitly or implicitly specifies an ordinal ranking $\succeq_v$ over all possible committees (that is, all possible subsets of $\C$). Given voter $v \in \N$ and two committees $S_1$ and $S_2$, we use the notation $S_2 \succ_v S_1$ to indicate the voter strictly prefers $S_2$ to $S_1$. We indicate weak preference by $S_2 \succeq_v S_1$. We assume preferences are {\em complete}, so that for every voter $v$ and every two committees $S_1$ and $S_2$, the voter either (weakly or strongly) prefers $S_1$ to $S_2$, or \emph{vice versa}; and the preferences are {\em transitive}, so that for every voter $v$, if she prefers committee $S_1$ to $S_2$ and prefers $S_2$ to $S_3$, then she prefers $S_1$ to $S_3$, where the preference is strict if at least one of the two preferences is strict. The only additional condition we impose on the preferences is the following:
\begin{quote}
{\bf Monotonicity}: If $S_1 \subseteq S_2$, then $S_2 \succeq_v S_1$ for all $v \in \N$.
\end{quote}

In Section~\ref{sec:special}, we will show that the facility location and participatory budgeting settings described above, as well as several other problems in social choice and network design, are special cases of committee selection.

\paragraph{Fairness via Stability.} The notion of {\em group fairness} or {\em proportionality} is a central objective in committee selection. The general idea appeared in literature more than a century ago~\cite{Droop} and various incarnations of this notion have gained significant attention recently~\cite{Chamberlain,Monroe,Brams2007,Brill,Sanchez,Fain2016,FainMS18,PJR2018}. Here, each group of voters should feel that their preferences are sufficiently respected, so that they are not incentivized to deviate and choose an alternative committee of proportionally smaller weight. In the common scenario that we do not know beforehand the exact nature of the demographic coalitions, we adopt the robust solution concept which requires the committee to be  {\em agnostic} to any potential subset of voters deviating.

Formally, we study fairness via the notion of {\em core stability} from economics literature~\cite{lindahlCore,scarfCore,coreConjectureCounter,LindahlPaper,SamuelsonPaper}. This uses the notion of {\em pairwise score} defined below.
\begin{define} [Pairwise Score]
\rm
Given two committees $S_1, S_2 \subseteq \C$, the \emph{pairwise score} of $S_2$ over $S_1$ is the number of voters who strictly prefer $S_2$ to $S_1$: $V(S_1,S_2) := | \{ v \in \N \ | \ S_2 \succ_v S_1\}|$.
\end{define}

Given the above definition, we are now ready to define fairness via core stability as follows.
\begin{define}[Stable Committees (or the Core)]
\label{def:stable}
\rm
Given a committee $S \subseteq \C$ of weight at most $K$, the weight limit, a committee $S' \subseteq \C$ of weight $K'$ blocks $S$ iff $V(S,S') \ge \frac{K'}{K} \cdot n$. A committee $S$ is \emph{stable} (or lies in the {\em core}) if no committee $S'$ blocks it.
\end{define}

In other words, for any $\beta \in (0,1]$ and any $\beta n$ voters, there should not be another committee of weight at most $\beta K$, so that all these $\beta n$ voters are strictly better off with the new committee. It is easy to check that a stable outcome is (weakly) Pareto-optimal among committees of weight at most $K$, by considering the deviating coalition of all the voters. Furthermore, for every coalition of voters, a stable committee is also Pareto-optimal relative to committees with proportionally scaled-down weight.

In economics, this notion of core stability can be justified with {\em fair taxation}~\cite{lindahlCore,LindahlPaper,SamuelsonPaper}: Each voter has an endowment of $\frac{K}{n}$, so the society together has a budget of $K$. Candidate $i$ costs $s_i$, and we select a committee whose weight is at most $K$. If no subset of voters (blocking committee) with size $\beta n$ can deviate and use their endowment ($\beta K$) to purchase an alternative committee, then the committee is said to be stable.

\subsection{Motivating Examples: Fair Combinatorial Selection}
\label{sec:special}
Our results hold for any monotone purely ordinal preference structure over committees. As such, it models a wide range of combinatorial selection problems that have a rich history in social choice, network design, and related domains.

\begin{description}
\item[{\sc Participatory Budgeting.}]  This models the civic budgeting application described above. 
Each candidate is a public project, and its weight $s_i$ equals its cost.
Voter $v$ has utility $u_{iv}$ for project $i \in \C$. The value $K$ is the total budget available to the city.  The utility of the voter for committee $S$ is $u_v(S) = \sum_{i \in S} u_{iv}$, and the voter prefers committees that provide her higher utility. This can be generalized to utility functions that capture  complements and substitutes.
\item[{\sc Approval Set.}] This is a special case of the setting described above that has been widely studied in multi-winner election literature. In this model~\cite{Brams2007,Brill,Sanchez,PJR2018}, we assume each $s_i = 1$. Each voter $v$ specifies an {\em approval set} $A_v \subseteq \C$ of candidates. Given two committees $S_1$ and $S_2$, $S_1 \succ_v S_2$ iff $|S_1 \cap A_v| > |S_2 \cap A_v|$, {\em i.e.}, the voter prefers committees in which she has more approved candidates.
\item[{\sc Ranking.}] In this model~\cite{ElkindFSS17}, each candidate has unit weight. Each voter $v$ has a preference ordering over candidates in $\C$. In this case, $S_1 \succ_v S_2$ iff $v$'s favorite candidate in $S_1$ is ranked higher (in her preference ordering) than her favorite candidate in $S_2$. 
\item[{\sc Facility Location.}] This is a special case of {\sc ranking} that is motivated by the problem of locating public facilities described above, and was recently considered by~\cite{ChenFLM19}. Here, the preferences over candidates are dictated by distances in an underlying metric space. Formally, there is a metric space $d$ over $\C \cup \N$. Each location in $\C$ is a potential facility. Given a subset $S \subseteq \C$ of $K$ locations, the cost of voter $v$ is the distance to the closest facility in $S$. A voter prefers $S_2$ to $S_1$ if it incurs smaller cost in the former than in the latter. 
\item[{\sc Network Design.}] The committee selection problem also models selection problems in network design and combinatorial optimization. For instance, consider the {\sc non-uniform  buy-at-bulk} network design problem~\cite{MeyersonMP00,ChekuriHKS06}. We are given a multigraph $G(V,E)$ with a sink node $s$. Each edge $e \in E$ as cost $c_e$ and length $\ell_e$.  The length is typically decreasing with cost, since a more expensive edge would model a faster road or higher capacity network cable, which would reduce time it takes to traverse that edge. Between any pair $u,v$ of vertices, we assume there is an edge $e = (u,v)$ of cost $c_e = 0$; let $M$ denote the subset of edges with cost $0$. A committee is a subset $S$ of edges with non-zero cost that are provisioned, and its weight is the total cost of its edges. Given committees $S$ and $S'$, an end-node $v$ prefers $S$ to $S'$ if the length of the shortest path (according to the edge lengths $\ell$) from $v$ to $s$ in the subgraph $G(V, S \cup M)$ is smaller in the subgraph $G(V, S' \cup M)$. 
\end{description}

\subsection{Approximate Stability and Main Result}
As shown in~\cite{ChengJMW19,FainMS18}, there are simple instances with cyclic preferences where stable committees may not exist; we present such an instance in Appendix~\ref{app:lb}. 
This motivates us to consider an approximate notion of stability.

\begin{define}[$c$-Approximately Stable Committees]
\label{def:approx-stable}
\rm
Given a parameter $c \ge 1$, and a committee $S \subseteq \C$ of weight at most $K$, the weight limit, we say that a committee $S' \subseteq \C$ of weight $K'$ $c$-blocks $S$ iff $V(S,S') \ge c \cdot \frac{K'}{K} \cdot n$. A committee $S$ is $c$-\emph{approximately stable} if there are no committees $S'$ that $c$-block it.\footnote{Any $c$-approximately stable committee can trivially be modified to become Pareto-optimal while preserving the value of $c$. To see this, we simply find another committee that Pareto-dominates this committee. Therefore, Pareto-optimality comes for free in our setting.}
\end{define}

In the taxation interpretation, we scale down the endowment of each voter by the approximation factor $c$, so that she has an endowment of $\frac{1}{c} \cdot \frac{K}{n}$. If a subset of voters with size $\beta n$ deviates and uses their endowment to purchase an alternative committee, this committee has weight $\frac{\beta}{c} K$. Note that when $c = 1$, the solution is exactly stable. A larger $c$ ensures fewer coalitions deviate, and our goal is to find the minimum $c$ for which a $c$-approximately stable solution exists.  Theorem~\ref{thm:lb} (Appendix~\ref{app:lb}) shows that $c \ge 2 - \varepsilon$ for any constant $\varepsilon > 0$ even in the {\sc Ranking} setting.  

Our main result is the following general and somewhat surprising theorem that we prove in the main body of the paper.

\begin{theorem}
\label{thm:main}
For {\em any} monotone preference structure with $n$ voters and $m$ candidates, arbitrary weights and the cost-threshold $K$, a $32$-approximately stable committee of weight at most $K$ always exists.
\end{theorem}

It is worth noting that prior to our work, no non-trivial result was known for the existence of approximately stable committees even in the very special cases of {\sc Approval Set} and {\sc Facility Location} preferences described above.   



\subsection{Techniques and Other Results}
Our proof of Theorem~\ref{thm:main} proceeds by first constructing a lottery (or randomization) over committees of weight $K$ that is $2$-approximately stable, and iteratively rounding this solution. The first challenge is to define the appropriate notion of randomized stability. As discussed in Section~\ref{sec:related}, though stability when committee members are chosen fractionally is a classical concept, these notions require convex and continuous preferences over fractional allocations, and it is not clear how to relate them to deterministic (or integer) solutions that we desire.

\paragraph{Stable Lotteries.} We proceed via a different randomized notion of stability that was first defined in~\cite{ChengJMW19}. We define this notion next.  Given a weight $K$, we let $\Delta$ denote a distribution (or lottery) over committees of weight at most $K$.

\begin{define}[Stable Lotteries~\cite{ChengJMW19}]
\label{def:lottery}
\rm
A distribution (or lottery) $\Delta$ over committees of weight at most $K$ is said to be $c$-\emph{approximately stable} iff for all committees $S' \subseteq \C$ of weight $K'$, we have: $\E_{S \sim \Delta} \left[ V(S,S') \right] < c \cdot \frac{K'}{K} \cdot n$.
\end{define}

In~\cite{ChengJMW19}, it was shown that an exactly stable lottery under this definition exists for {\sc Approval Set} and {\sc Ranking} settings, via solving the dual formulation. However, it was not clear either how to extend this technique even to {\sc Participatory Budgeting} preferences, or what a stable lottery implied about deterministic stable committees that is our main focus here. In this paper, we resolve both these questions. As our first contribution, in Section~\ref{sec:randomized}, we show the following.

\begin{theorem}
\label{thm:random}
 For any weight $K$ and  {\em all} monotone preferences, a $2$-approximately stable lottery over committees of weight at most $K$ always exists.
\end{theorem}

The proof of the above theorem builds on the proof of the exactly stable lottery for {\sc Ranking} instances in~\cite{ChengJMW19}. The duality proof in~\cite{ChengJMW19} constructed a primal lottery by sequentially rounding candidates based on their marginal probability in the dual solution, while the current proof constructs the primal lottery by sequentially rounding the committees in the dual solution directly. This allows us to develop a simple proof for all monotone preference structures; however, unlike~\cite{ChengJMW19}, our lottery is only approximately stable.

Once we construct this lottery, our main contribution (Section~\ref{sec:det}) is rounding it to show Theorem~\ref{thm:main}.  The randomized stability condition implies the existence of a committee that satisfies a certain fraction of voters simultaneously, in the sense that it lies not too far down the preference ordering of these voters. We iteratively eliminate such voters and re-compute the lottery, with the non-trivial aspect being to ensure that this process preserves approximate stability.

\medskip
In Section~\ref{sec:exten}, we show that our results extend to the more general setting where $w(S)$ is a subadditive set function, and also to the setting where there are multiple weight constraints. We also discuss some settings in which an approximately stable committee can be efficiently computed.

\paragraph{Exactly Stable Lotteries.} When considering lotteries, we haven't been able to find an instance of a monotone preference structure where an exactly stable solution does not exist. The loss of factor of $2$ in Theorem~\ref{thm:random} seems to be an artifact of our analysis. Indeed, in Appendix~\ref{sec:three}, we show a different way of constructing the dual solution that leads to the following result. We conjecture that this results extends to all $K$, and we discuss this and other open questions in Section~\ref{sec:open}. 

\begin{theorem} [Proved in Appendix~\ref{sec:three}]
\label{thm:small}
For unit-weight candidates and any number of voters with arbitrary monotone preferences, when $K \in \{1, 2, 3\}$, an {\em exactly} stable lottery always exists.
\end{theorem}


\subsection{Related Work}
\label{sec:related}
Committee selection is omnipresent in political and economic activities of a society: We see it in parliamentary elections, in group-hiring processes, and in participatory budgeting. Recent work in social choice~\cite{Procaccia2008,Meir2008,Lu2011,Brill,Sanchez,PJR2018} has extensively studied the properties of committee selection rules and established axiomatization in this field. Furthermore, group fairness in committee selection arises in areas outside social choice: In a shared-cache system with multiple users, consider the problem of deciding which parts of the data to keep in the cache that has only limited storage~\cite{ROBUS,Psomas}. Users gain utility from their data being cached. We can model this as committee selection where each atomic piece of data corresponds to a candidate. In this context, a fair caching policy provides proportional speedup to each user.

\medskip
We now compare our stability notions with some closely related notions in literature. This will place our technical work in context.

\paragraph{The Lindahl Equilibrium.}
The notion of stability is the same as that of the \emph{core} in cooperative game theory. Scarf~\cite{scarfCore} first phrased it in game-theoretic terms, and it has been extensively studied in public-good settings~\cite{LindahlPaper,SamuelsonPaper,lindahlCore,coreConjectureCounter,Fain2016}. Much of this literature considers convex and continuous preferences, which in our setting implies convex preferences over fractional allocations (that is, when candidates can be chosen fractionally). The seminal work of Foley~\cite{lindahlCore} considers the Lindahl market equilibrium. In this equilibrium, each candidate is assigned a per-voter price. If the voters choose their utility maximizing allocation subject to spending a dollar, then (1) they all choose the same fractional outcome; and (2) for each chosen candidate, the total money collected pays for that candidate. It is shown via a fixed point argument that such an equilibrium pricing always exists when fractional allocations are allowed, and this outcome lies in the core. Though this existence result is very general, it needs the preferences over fractional outcomes to be convex and continuous.

For instance, in the case of {\sc facility location}, a fractional allocation $\vec{y}$ satisfies $\sum_{i \in \C} y_i \leq K$ and $y \in [0, 1]$. One possible convex disutility of voter $v$ for $\vec{y}$ is
\[
C_v(\vec{y}) = \min \left\{ \sum_i d_{iv} x_{iv} \ \middle| \ \sum_i x_{iv} = 1; \ x_{iv} \le y_i \ \forall i \right\}.
\]
Though Foley's result shows there exists an allocation $\vec{y}$ that is a core outcome, it is not clear (a) how to compute this fractional solution efficiently; and (b) more importantly, how to round this allocation to an approximately stable integer solution. The difficulty in rounding is because we cannot relax the distances when considering when a voter can deviate; indeed, if we could relax distances, the problem becomes very different, and  there is an approximately stable solution (that only relaxes distances and not the size of the deviating coalition) via a simple greedy algorithm~\cite{ChenFLM19}.

This motivates using the new notion of randomized stability, where a deterministic outcome is first drawn from the lottery, and subsequently the voters who see higher utility deviate. This notion does not correspond to underlying convex preferences over the space of lotteries; however, as we show, a stable lottery can now be converted to an approximately stable committee. Furthermore, for {\sc facility location} and more generally, {\sc Ranking},  this stable lottery can be efficiently computed~\cite{ChengJMW19}, while we do not know how to compute the Lindahl equilibrium efficiently.

\paragraph{Nash Welfare and its Variants.} There is extensive work (see~\cite{ElkindFSS17}) on voting rules where we construct a score $\sigma_v(S)$ for each voter and committee, and choose the committee that maximizes $\sum_v \sigma_v(S)$. For instance,  for {\sc Approval Set} preferences, the classic Proportional Approval Voting (PAV) method that dates back more than a century~\cite{Thiele}, assigns $\sigma_v(S) \approx  \log(1 + |A_v \cap S|)$.  More generally, the {\em Nash Welfare} objective~\cite{Portioning,Fain2016,Thiele,PJR2018,FairKnapsack,FainMS18} assigns score $\sigma_v(S) = \log(u_v(S))$, where $u_v(S)$ is the utility of the voter for committee $S$. These methods compute a stable solution when the utility of voters in a deviating coalition is scaled down. This requires either knowing or  imputing cardinal utility functions for voters (and does not work with disutilities), and is otherwise incomparable to the more widely studied and classical notion of core stability that we consider. where the committee size on deviation is scaled down and the utilities of voters are unchanged.  Further, for {\sc Approval set} preferences, the PAV method is no better than an $\Omega(\sqrt{K})$-approximation to a stable outcome~\cite{ChengJMW19}.

A recent line of work~\cite{Sanchez,PJR2018,Brill} has considered a special case of stability with {\sc Approval set} preferences, when the coalition that deviates is not arbitrary, but is cohesive in terms of preferences. They term this {\em Justified Representation}, with generalizations known both for {\sc Approval set} and to other preference structures~\cite{Aziz1,Aziz2}. For {\sc Approval Set}, it is shown that the PAV method and its variants achieve or closely approximate these notions of stability. However, as mentioned above, the PAV method do not approximate the core outcome, so that stability is very different in structure from Justified Representation and its variants. 

Finally, our approximation notion scales down the endowment of the deviating coalition by a factor of $c$. An alternative approach would have been to approximate the utility of the voters in the deviating coalition by a factor of $c$. In this model, constant approximations have been obtained for clustering~\cite{ChenFLM19} and for {\sc Approval Set}~\cite{Peters20}; the latter result uses the PAV method. However, these results require developing a different technique for each problem, while our approach has the advantage of being oblivious to the choice of cardinal utilities while leading to a unifying result for all preference structures.

\section{Existence of $2$-Approximately Stable Lotteries}
\label{sec:randomized}
\label{sec:random}
In this section, we consider choosing a stable lottery over committees of weight at most $K$. Recall the definition of stability in this setting from Definition~\ref{def:lottery}.

We take the approach in~\cite{ChengJMW19} and consider the dual formulation of selecting stable lotteries. The existence of a $c$-approximately stable lottery is equivalent to deciding:
\begin{equation}
\min_{\Delta} \max_{S'} \E_{S \sim \Delta}\left[V(S, S') - c \cdot \frac{w(S')}{K} \cdot n\right] < 0,\label{eq:primal_game}
\end{equation}
where $\Delta$ is a distribution (lottery) over committees of weight at most $K$. Viewing $\Delta$ as a mixed strategy over the ``defending'' committees and $S'$ as the ``attacking'' strategy, we treat (\ref{eq:primal_game}) as a zero-sum game. Duality (or the {\em min-max principle}) now allows us to swap the order of actions by allowing the attacker to use a mixed strategy. (\ref{eq:primal_game}) is thus equivalent to
\begin{equation}
\max_{\Delta_a} \min_{S_d : w(S_d) \leq K} \E_{S_a \sim \Delta_a}\left[V(S_d, S_a) - c \cdot \frac{w(S_a)}{K} \cdot n\right] < 0,\label{eq:dual_game}
\end{equation}
where $\Delta_a$ is a lottery over committees of weight at most $K$ chosen by the attacker. This dual view provides a convenient tool for showing the existence of approximately stable lotteries. The rest of the section is devoted to proving Theorem~\ref{thm:random}, that we restate here.

\begin{reptheorem}{thm:random}
For any value $K$ and  {\em all} monotone preferences, a $2$-approximately stable lottery over committees of weight at most $K$ always exists.
\end{reptheorem}

\subsection{Per-Voter Guarantee}
Assume we are given a lottery $\Delta_a$. If there is a committee $S_a$ in $\Delta_a$ with $w(S_a) > \frac{K}{2}$, then $2 \cdot \frac{w(S_a)}{K} \cdot n > n \ge V(S_d, S_a)$ for any $S_d$. This implies $V(S_d, S_a) - 2 \cdot \frac{w(S_a)}{K} \cdot n < 0$. Therefore, the attacker can remove these strategies from its lottery, and we can assume $\Delta_a$ only has committees with weight at most $\frac{K}{2}$.

Given any distribution $\Delta_a$ over committees of weight at most $\frac{K}{2}$, we need to show there is a defending committee $S_d$ with weight at most $K$, such that
$$
\E_{S_a \sim \Delta_a}\left[V(S_d, S_a) - 2 \cdot \frac{w(S_a)}{K} \cdot n\right] < 0.
$$

Suppose that the strategy $\Delta_a$ chooses $S_1$ with probability $\alpha_1$, committee $S_2$ with probability $\alpha_2$, \ldots, $S_t$ with $\alpha_t$, where $t = 2^{|\C|}$. Let
$$\beta = \E_{S_a \sim \Delta_a}\left[\frac{w(S_a)}{K}\right] = \frac{\sum_{i = 1}^t \alpha_i \cdot w(S_i)}{K}$$
be the ratio between the expected total weight of the attacking strategy and $K$, the allowable weight for the defending strategy.  We need to find an $S_d$ of weight at most $K$ so that:
\begin{equation}
 \E_{S_a \sim \Delta_a}\left[V(S_d, S_a) \right] < 2 \beta n. \label{eq:random_goal}
\end{equation}
We will construct a distribution $\Delta_d$ over committees $S_d$ that satisfies a stronger property:
\begin{equation}
 \Pr_{S_d \sim \Delta_d, S_a \sim \Delta_a}[S_a \succ_v S_d] < 2 \beta \qquad \forall v \in [n]. \label{eq:random_goal2}
 \end{equation}
Summing over all voters $v$ implies the existence of a lottery $\Delta_d$ satisfying (\ref{eq:random_goal}), and hence a deterministic committee $S_d$ satisfying the same. This will imply the theorem statement.

\subsection{Dependent Rounding}
Let $p_i = \min\left(1, \frac{\alpha_i}{2\beta}\right)$ for $i \in [t]$. We have:
\begin{itemize}
\item $p_i \in [0,1]$ for all $i \in [t]$; and
\item $\sum_{i \in [t]} p_i w(S_i) \le \frac{1}{2 \beta} \sum_{i \in [t]} \alpha_i w(S_i) = \frac{\beta K}{2 \beta} = \frac{K}{2}$.
\end{itemize}

We will construct the defending committee by including each attacker committee $S_i, i \in [t]$ with probability $p_i$; the details of which are below. We use the random variable $X_i$ to denote whether we include $S_i$ in our defending committee $S_d$, so that $S_d = \bigcup_{i \in [t] : X_i > 0} S_i$. We therefore obtain a distribution $\Delta_d$ over committees $S_d$.

We use the dependent rounding procedure in~\cite{byrka} to construct $\{X_i\}$ given the $\{p_i\}$. This has the following properties.
\begin{itemize}
\item (Almost-Integrality) For any realization of $X_i$'s, all but at most one of them takes value in $\{0, 1\}$ (the remaining one takes value in $[0, 1]$).
\item (Correct Marginals) $\E[X_i] = p_i$ for all $i \in [t]$.
\item (Preserved Weight) $\Pr\left[\sum_{i \in [t]} w(S_i) \cdot X_i \leq \frac{K}{2}\right] = 1$.
\item (Negative Correlation) $\forall T \subseteq [t]$, $\E\left[\prod_{i \in T} (1 - X_i)\right] \leq \prod_{i \in T}\E[1 - X_i] = \prod_{i \in T} (1 - p_i)$.
\end{itemize}

To have full integrality instead of the almost-integrality, in any realization, we include $S_i$ in our $S_d$ as long as $X_i > 0$ (instead of only fully including it when $X_i = 1$). Since we assumed $w(S_i) \le \frac{K}{2}$ for all $i \in [t]$, using the almost-integrality and preserved-weight conditions, for any realization of $\{X_i\}$, the weight of the resulting $S_d$ satisfies
\[
w(S_d) \le \sum_{i \in [t] : X_i = 1} w(S_i) + \sum_{i \in [t] : X_i \in (0, 1)} w(S_i) \le \frac{K}{2} + \frac{K}{2} = K.
\]

\subsection{Analysis}
Fix a voter $v$. W.l.o.g. assume her preference over the sets in $\Delta_a$ are $S_1 \succeq_v S_2 \succeq_v \cdots \succeq_v S_t$.
\begin{align*}
\Pr_{S_d \sim \Delta_d, S_a \sim \Delta_a}[S_a \succ_v S_d] &\leq \sum_{i \in [t]} \Pr_{S_a \sim \Delta_a}[S_a = S_i \text{ and } X_1 = X_2 = \cdots = X_{i} = 0]\\
& =  \sum_{i \in [t]} \alpha_i \cdot \Pr_{S_a \sim \Delta_a}[X_1 = X_2 = \cdots = X_{i} = 0]\\
&\leq \sum_{i \in [t]} \alpha_i \cdot \prod_{j = 1}^i (1-p_j)\\
&=  \sum_{i \in [t]} 2\beta \cdot p_i \cdot \prod_{j = 1}^i (1-p_j)  <  2\beta.
\end{align*}

Here, the first step follows because when the adversary chooses set $S_i$, it only beats $S_d$ if $S_d$ included none of $S_1, S_2, \ldots, S_i$. Here, we are using the monotonicity of the preference structure: Since $S_j \succeq_v S_i$ for $j \le i$, this implies $S_d \succeq_v S_i$  when $S_j \subseteq S_d$. The second step follows since the realization of the adversary's lottery is independent of that of the defender, and since $\alpha_i = \Pr[S_a = S_i]$. The third step follows by the negative correlation property of $\{X_j\}$. To see the fourth step, note that if $\alpha_i \ge 2 \beta$, then $p_i = 1$, so that $\alpha_i  \cdot \prod_{j = 1}^i (1-p_j) = 0 = 2\beta \cdot p_i \cdot \prod_{j = 1}^i (1-p_j)$. Otherwise, $\alpha_i = 2 \beta p_i$.

To see the final inequality, note that $\sum_{i \in [t]} p_i \prod_{j < i} (1-p_j)$ is the probability of the following stopping process picking some set: Pick $S_1$ with probability $p_1$; if not, pick $S_2$ with probability $p_2$, and so on. Therefore,  $\sum_{i \in [t]} p_i \prod_{j < i} (1-p_j) \le 1$. 
Thus $\sum_{i \in [t]} p_i \prod_{j \le i} (1-p_j) < 1$. This proves (\ref{eq:random_goal2}), and hence Theorem~\ref{thm:random}.

\newcommand{\V}{\mathcal{V}}
\newcommand{\G}{\mathcal{G}}
\newcommand{\B}{\mathcal{B}}
\newcommand{\W}{\mathcal{W}}
\renewcommand{\H}{\mathcal{H}}
\section{Existence of Approximately Stable Committees}
\label{sec:det}
In this section, we show that a $O(1)$-approximately stable committee always exists. We show this by iteratively rounding the lottery constructed in Section~\ref{sec:randomized}. We first restate Theorem~\ref{thm:main}.

\begin{reptheorem}{thm:main}
For {\em any} monotone preference structure over any number $n$ of voters, and $m$ of candidates with arbitrary weights, and any weight $K$, a $32$-approximately stable committee of weight at most $K$ always exists.
\end{reptheorem}

For the proof, fix a deviating committee $S_a$ of weight $w(S_a)$. Suppose our final committee is $T$ of weight at most $K$. Our goal is to show that:  $V(T,S_a) < 32 \cdot \frac{w(S_a)}{K} \cdot n.$

\subsection{Good and Bad Committees}
Throughout the proof, we fix two constants $0 < \beta \le \alpha < 1$, whose choice will be determined at the very end. To begin, we define a subroutine that returns a $2$-approximately stable lottery (via Theorem~\ref{thm:random}) for any subset of voters, and any committee size.

\begin{define}
\rm
Given candidate set $[m]$, 
voter set $\V'$, and 
committee size $K'$, let $\textsc{Lottery}(\V', K')$ return a lottery $\Delta$ over committees of weight at most $K'$ that is $2$-approximately stable for the set of voters $\V'$. Similarly, let
$ V_{\V'}(S,S_a) = \left| \left\{ v \in \V' \ \middle| \ S_a \succ_v S \right \} \right|.$

\end{define}

Let $x_i$ be the probability that $\Delta$ includes committee $S_i$.

\begin{define}
\rm
Given a voter $v$, we define the set of {\em good} and {\em bad} committees relative to $\Delta$, $\G_v(\Delta)$ and $\B_v(\Delta)$ respectively,  as follows:
\rm
\[
\G_v(\Delta) = \left\{ S \subseteq \C \ \middle|\  \sum_{S_i \succeq_v S} x_i \leq 1 - \beta \right\} \qquad \mbox{and} \qquad  \B_v(\Delta) = \left\{ S \subseteq \C \ \middle|\  \sum_{S_i \preceq_v S} x_i \leq \beta \right\}.
\]
\end{define}

The idea is that the good committees appear sufficiently high up in $v$'s ranking, while the bad committees are lower down in the ranking. The notion of high and low is relative to the probability mass $\Delta$. The following lemma is immediate.

\begin{claim}
\label{claim1}
If $S \notin \B_v(\Delta)$, then  $S_a \succ_v S$ only if $S_a \in \G_v(\Delta)$. 
\end{claim}
\begin{proof}
If $S \notin \B_v(\Delta)$, then $\sum_{S_i \preceq_v S} x_i >  \beta$. This implies $\sum_{S_i \succ_v S} x_i < 1 - \beta$. Since $S_a \succ_v S$, we have $\sum_{S_i \succeq_v S_a} x_i < 1 - \beta$, so that $S_a \in \G_v(\Delta)$.
\end{proof}

The next lemma implies that (a) Any committee $S_a$ cannot lie in too many good sets relative to its weight; and (b) There is some committee (with non-zero support in $\Delta$) that does not lie in more than a constant fraction of the bad sets. The previous claim rules out the possibility where too many voters prefer $S_a$ to such a committee, relative to the weight of $S_a$, which will be crucial for the algorithm we subsequently design.

\begin{lemma}
\label{lem:satisfied}
Given $\Delta = $ {\sc Lottery}$(\V',K')$, we have the following upper and lower bounds:
 \begin{enumerate}
 \item For all committees $S_a$, we have
 \[
 \left|\left\{v \in \V' \ \middle| \ S_a \in \G_v(\Delta)  \right\} \right| < \frac{2}{\beta} \cdot \frac{w(S_a)}{K'} \cdot |\V'|.
 \]
\item There exists  $S$ with non-zero support in $\Delta$ such that
\[
\left| \left\{v \in \V' \ \middle| \ S \notin \B_v(\Delta)  \right\} \right| \ge (1-\beta) \cdot |\V'|.
\]
\end{enumerate}
\end{lemma}
\begin{proof}
To see the first part, for any committee $S_a$,  we have: $1 - \sum_{S_i \succeq_v S_a} x_i =  \Pr_{S_i \sim \Delta}[S_i \prec_v S_a].$ Summing over $v \in \V'$,
\[
\sum_{v \in \V'} \left(1 - \sum_{S_i \succeq_v S_a} x_i\right) = \E_{S_i \sim \Delta} \left[V_{\V'}(S_i, S_a) \right] < \frac{2w(S_a)}{K'} \cdot |\V'|,
\]
where the inequality comes from the fact that $\Delta$ is $2$-approximately stable. Thus there are fewer than $\frac{2}{\beta} \cdot \frac{w(S_a)}{K'} \cdot |\V'|$ voters $v \in \V'$ with $\sum_{S_i \succeq_v S_a} x_i \leq 1 - \beta$, which is necessary for $S_a \in \G_v(\Delta)$.

To see the second part, suppose $S \sim \Delta$. Then for each $v \in \V'$, since $\sum_{S_i \notin \B_v(\Delta)} x_i \geq 1 - \beta$, we have: $\Pr \left[S \notin \B_v(\Delta) \right] \ge 1 - \beta$. Therefore, the expected number (over the choice $S \sim \Delta$) of $v$ such that $S \notin \B_v(\Delta)$  is at least $(1-\beta) \cdot |\V'|$. Therefore, there exists an $S$ for which the claim holds.
\end{proof}


\subsection{Algorithm}
Algorithm~\ref{alg1} shows our full procedure. The main idea is the following: If we pick a committee $S$ that does not lie in $\B_v(\Delta)$ for most voters $v$, then by Claim~\ref{claim1}, $S_a$ is forced to lie in $\G_v(\Delta)$ for these voters if $S_a$ beats $S$. But since $\Delta$ is $2$-approximately stable, by Lemma~\ref{lem:satisfied}, there are only a small number of $v$ where $S_a$ can lie in $\G_v(\Delta)$. We can therefore remove these set of voters for whom $S \notin \B_v(\Delta)$, since $S$ makes sure no $S_a$ can capture too many of these voters. This reduces the number of voters by a constant factor. For the remaining voters, we recursively find another committee of smaller (but not too much smaller) weight, which reduces the number of voters by another constant factor; and so on. The key point is that the total weight of all these committees is a geometric sequence, {\em and} the number of voters who can be captured by $S_a$ in each round is also a geometric sequence, showing a constant-approximately stable solution.

\begin{algorithm}[htbp]
\caption{Iterated Rounding}
\label{alg1}
\begin{algorithmic}[1]
	\State  $t \gets 0$;  $\V^{(0)} \gets [n]$; $T^{(0)} \gets \varnothing$;  $K^{(0)} \gets (1-\alpha) K$.
	\While{$\V^{(t)} \neq \varnothing$}
		\State $\Delta^{(t)} \gets $ {\sc Lottery}$(\V^{(t)},K^{(t)})$.
		\State \label{state1} Let $S^{(t)}$ be any committee such that $\left| \left\{v \in \V^{(t)} \ \middle| \ S^{(t)} \notin \B_v\left(\Delta^{(t)}\right)  \right\} \right| \ge (1-\beta) \cdot |\V^{(t)}|$.
		\State $\W^{(t)} \gets  \left\{v \in \V^{(t)} \ \middle| \ S^{(t)} \notin \B_v\left(\Delta^{(t)}\right)  \right\}$.
		\State $\V^{(t+1)} \gets V^{(t)} \setminus \W^{(t)}$.
		\State $T^{(t+1)} \gets T^{(t)} \cup S^{(t)}$.
		\State $K^{(t+1)} \gets \alpha K^{(t)}$.
		\State $t \gets t+1$.
	\EndWhile
	\State \Return $T^f \gets T^{(t)}$.
\end{algorithmic}
\end{algorithm}

\subsection{Analysis}
Line \ref{state1} in Algorithm~\ref{alg1} is correct by Lemma~\ref{lem:satisfied}. We next bound the weight of the final set.

\begin{lemma}
$w(T^f) \le K$.
\end{lemma}
\begin{proof}
$w(T^f) \le \sum_{t \ge 1} w(S^{(t)}) \le \sum_t \alpha^{t-1} (1-\alpha) K \le K$.
\end{proof}

We finally show that the resulting set is approximately stable, completing the proof of Theorem~\ref{thm:main}.

\begin{lemma}
When $\alpha = \frac{1}{2}$ and $\beta = \frac{1}{4}$, then $T^f$ is a $32$-approximately stable committee of weight at most $K$.
\end{lemma}
\begin{proof}
Given $S_a$, since $\V = [n] = \bigcup_{t \ge 1} \W^{(t)}$ and $T^f = \bigcup_{t\ge 1} S^{(t)}$, using monotonicity, we have:
\[
V\left(T^f, S_a\right) \le \sum_{t \ge 1} V_{\W^{(t)}}\left(T^f,S_a\right) \le  \sum_{t \ge 1} V_{\W^{(t)}}\left(S^{(t)},S_a\right).
\]
Since $S^{(t)} \notin \B_v\left(\Delta^{(t)}\right)$ for $v \in \W^{(t)}$, by Claim~\ref{claim1}, $S_a \succ_v S^{(t)}$ only if $S_a \in \G_v\left(\Delta^{(t)}\right)$. By Lemma~\ref{lem:satisfied},
\[
\left|\left\{v \in \V^{(t)} \ \middle| \ S_a \in \G_v\left(\Delta^{(t)}\right)  \right\} \right| < \frac{2}{\beta} \cdot \frac{w(S_a)}{K^{(t)}} \cdot |\V^{(t)}|.
\]
Note now that $|\V^{(t)}| \le \beta |\V^{(t-1)}|$, so that $ |\V^{(t+1)}|  \le \beta^t n$. Furthermore, $K^{(t)} = \alpha K^{(t-1)}$, so that $K^{(t+1)} = \alpha^t (1-\alpha) K$. Therefore,
\begin{eqnarray*}
 V_{\W^{(t)}}\left(S^{(t)},S_a\right) & \le &  \left|\left\{v \in \V^{(t)} \ \middle| \ S_a \in \G_v\left(\Delta^{(t)}\right) \right\} \right| \\
 & < & \frac{2}{\beta} \cdot \frac{w(S_a)}{K^{(t)}} \cdot |\V^{(t)}|  \le  \left(\frac{\beta}{\alpha} \right)^{t-1} \cdot \frac{2}{(1-\alpha)\beta} \cdot \frac{w(S_a)}{K} \cdot n. \\
\end{eqnarray*}
Summing over all $t$, we have
$$  V\left(T^f, S_a\right) \le  \sum_{t \ge 1} V_{\W^{(t)}}\left(S^{(t)},S_a\right) < \frac{2 \alpha}{\beta(1-\alpha)(\alpha-\beta)} \cdot  \frac{w(S_a)}{K} \cdot n. $$
This is minimized when $\beta = \frac{\alpha}{2}$. Setting $\alpha = \frac{1}{2}$, this is at most $32 \cdot \frac{w(S_a)}{K} \cdot n $, completing the proof.
\end{proof}

\section{Extensions}
\label{sec:exten}
We now present some extensions of the above results to the setting where weights are subadditive, and there are multiple weight constraints. We also show settings in which the Algorithm~\ref{alg1} has an efficient implementation.

\subsection{General Weight Functions}
\noindent{\bf Subadditive Weights.} A careful reader may have observed that the only property we have used in the proofs of Theorems~\ref{thm:main} and~\ref{thm:random} is that $w(S_1 \cup S_2) \le w(S_1) + w(S_2)$.  Therefore, we have: 

\begin{corollary}
There is a $32$-approximately stable committee for any subadditive weight function $w(S)$ over committees, and any monotone preferences of the voters.
\end{corollary}

\noindent {\bf Multiple Constraints.} We note that Theorems~\ref{thm:main} and~\ref{thm:random} naturally extend to the following setting with multiple weight constraints. In the multi-constraint setting, there are $Q$ types of resources, and the weight limit of the $i$-th resource is $K_i$. Given a subadditive weight $w_j(S)$ for committee $S$ and resource $j$, we select a committee $S$ so that all $Q$ constraints $w_j(S) \leq K_j$ are respected. A coalition $\mathcal{V}' \subseteq \mathcal{V}$ should only have access to $\frac{|\mathcal{V}'|}{|\mathcal{V}|} \cdot K_j$ amount of resource $j$.
\begin{define}[Stable Committees with Multiple Constraints]
\label{def:stable_multi}
\rm
Given a committee $S \subseteq \C$ of weight at most $(K_1, K_2, \ldots, K_Q)$, a committee $S' \subseteq \C$ of weight $(K_1', K_2', \ldots, K_Q')$ blocks $S$ iff $V(S,S') \ge \frac{K_j'}{K_j} \cdot n$ for all $j \in [Q]$. A committee $S$ is \emph{stable} if no committee $S'$ blocks it.
\end{define}
Notions of $c$-approximately stable committees and lotteries can be similarly generalized. By normalizing the weights, we can assume the cost limits are $K_1 = K_2 = \cdots = K_Q = K$. Redefine the weight of a committee $S$ to be the maximum weight across the resources, $\emph{i.e.}$, $w(S) := \max_{j \in [Q]} w_j(S)$. This weight function is also subadditive. Further, any ($c$-approximately) stable solution in the new single-resource instance would also be $c$-approximately stable in the original multi-resource one: It is straightforward to verify all $Q$ constraints are satisfied in the original setting and the no-deviation requirements are exactly the same in both settings. We therefore have:

\begin{corollary}
There is a $32$-approximately stable committee in the setting with $Q \ge 1$ resources.
\end{corollary}

\subsection{Running Time}
Our main result above is that of existence of approximately stable committees.  If preferences are arbitrary, then we can find this solution by brute-force calculation of $V(S,S')$ for all pairs of feasible committees $(S,S')$, which takes time exponential in $K$. Our algorithm has comparable running time, and the bottleneck is constructing a stable lottery efficiently. Indeed, Algorithm~\ref{alg1} runs in poly$(m,n)$ time if we can find an approximately stable lottery with polynomial size support in polynomial time. Achieving this for {\sc Approval Set} or {\sc Participatory Budget} setting is still an open question. We now present some settings where a more efficient implementation is possible.

We first define the following notion of $(c,L)$-approximately stable committee, generalizing the notion defined in~\cite{ChengJMW19} to arbitrary weights.

\begin{define}[$(c,L)$-Approximately Stable Committee]
\rm
 A committee $S \subseteq \C$ of weight at most an integer value $K$ is $(c,L)$-\emph{approximately stable} for $1 \le L \le K$ if there is no committee $S'$ with {\em at most} $L$ candidates such that $V(S,S') \ge c \cdot \frac{w(S')}{K} \cdot n$.
\end{define}

In the {\sc Approval set} setting, a $(1,1)$-stable committee is exactly a committee that satisfies Justified Representation~\cite{Brill}, which in itself is a non-trivial property. If we restrict the attacking committees $S_a$ to have at most $L$ candidates, the number of such committees is $O\left(m^L\right)$. It is now easy to show that the $(2+\varepsilon,L)$-approximately stable lottery in Section~\ref{sec:random} can be computed in time poly$\left(m^L,n,\frac{1}{\varepsilon}\right)$ via the multiplicative weight update (MWM) method for solving zero-sum games. The idea is that given a distribution over $S_a$, the defending strategy involves dependent rounding over this distribution and hence runs in poly$\left(m^L\right)$ time. The number of rounds of MWM will be polynomial in the number of attacking strategies and $\frac{1}{\varepsilon}$, and the resulting distribution over defending strategies will have a support of size poly$\left(m^L,\frac{1}{\varepsilon}\right)$. This implies Algorithm~\ref{alg1} is efficient for constant $L$.

\begin{corollary}
For any $1 \le L \le K$, a $(32+\varepsilon,L)$-approximately stable committee can be computed in poly$\left(m^L,n,\frac{1}{\varepsilon}\right)$ time.
\end{corollary}

In the {\sc Ranking} setting with additive weights, it is easy to observe that a committee is $c$-approximately stable iff it is $(c,1)$-approximately stable.  This directly implies the following.  

\begin{corollary}
\label{cor1}
For sufficiently small constant $\varepsilon > 0$, a  $(32+\varepsilon)$-approximately stable committee for {\sc Ranking} and {\sc Facility Location} preferences, even when candidates have arbitrary (additive) weights, can be computed in poly$\left(m,n,\frac{1}{\varepsilon}\right)$ time.
\end{corollary}
This can be improved to a $(16+\varepsilon)$-approximation when the candidates are unweighted, since an exactly stable lottery exists in this setting~\cite{ChengJMW19}. We  emphasize that prior to this, no such result was known, even for existence of an approximately stable solution for these preferences.


\section{Open Questions}
\label{sec:open}
There are several challenging open questions,  both for existence and computation. 

\begin{itemize}
\item Does an exactly stable lottery always exist for all monotone preference structures? We show its existence for committees of weight $K \le 3$ and unit-weight candidates in Appendix~\ref{sec:three}. Though there is some intuition that the problem in Section~\ref{sec:randomized} resembles fractional knapsack and hence must overflow the knapsack while rounding, our proof for $K = 3$ shows that this intuition is misleading. Indeed, our proof uses a different rounding procedure than standard dependent rounding, and it is an open question whether such a procedure always exists.
\item Does an exactly stable committee exist for the {\sc Approval set} setting? For this specific setting, no counterexample to exact stability is known. 
\item For {\sc Approval Set} or {\sc Participatory budgeting}, can an approximately stable lottery (and hence a deterministic approximation) be efficiently computed? Unlike \textsc{ranking} and \textsc{Facility Location} settings, it is possible that a solution is not reasonably approximately stable, but no deviating coalition is small. (See, \emph{e.g.}, the $\Omega(\sqrt{K})$ lower bound example for PAV rules in~\cite{ChengJMW19}.) On the other hand, though there are exponentially many committees, the preference structure in these settings is simple and we cannot rule out polynomial time algorithms.
\item Does a $2$-approximately stable committee always exist for any monotone preference structure? We conjecture that Theorem~\ref{thm:lb} in Appendix~\ref{app:lb} is in fact tight, and the factor of $32$ in Theorem~\ref{thm:main} can possibly be lowered by other approaches.
\end{itemize}

\section*{Acknowledgments}
We thank Yu Cheng for many helpful discussions. This work was done while Zhihao Jiang was visiting Duke University. This work is supported by NSF grants CCF-1408784, CCF-1637397, and IIS-1447554; ONR award N00014-19-1-2268; and research awards from Adobe and Facebook.

\bibliographystyle{abbrv}
\bibliography{abb,refs}

\appendix
\section{Lower Bound on Approximation}
\label{app:lb}
It is relatively easy to construct a {\sc Participatory Budgeting} preference profile where a $(2 - \varepsilon)$-approximately stable committee does not exist. This instance has cyclic preferences. There are $m$ candidates $\{c_i\}_{i \in [m]}$ of unit weight, and $n = m$ voters $\{v_i\}_{i \in [n]}$. Let $K = 2 - \frac{\varepsilon}{2}$. The preference of $v_i$ is:
\[
c_i \succ c_{i + 1} \succ \cdots \succ c_m \succ c_1 \succ \cdots \succ c_{i - 1}.
\]
Any feasible committee is some single candidate $c_i$, but all voters except $v_i$ can deviate and choose $c_{i - 1}$ (or $c_m$ if $i = 1$). Therefore, the approximation ratio is at least $\frac{m - 1}{n} \cdot K = \frac{m - 1}{m} \cdot \left(2 - \frac{\varepsilon}{2}\right) > 2 - \varepsilon$ when $m$ is large enough.

\medskip
We now strengthen this example to show that even for the {\sc Ranking} setting with unit-weight candidates and integral committee weight $K$, there exist instances where a $(2 - \varepsilon)$-approximately stable committee does not exist.
\begin{theorem}
\label{thm:lb}
In the unweighted \textsc{Ranking} setting, $(2 - \varepsilon)$-approximately stable deterministic committees of integral size $K$ may not exist for any constant $\varepsilon > 0$.
\end{theorem}
\begin{proof}
For any positive integer $r \geq 2$ and $\ell$, we construct the following instance with $n = r \cdot \ell$ voters and $m = r \cdot \ell$ candidates. We view the candidates $\{c_{i, j}\}_{i \in [r], j \in [\ell]}$ as a matrix with $r$ rows and $\ell$ columns. The voters are $\{v_{i, j}\}_{i \in [r], j \in [\ell]}$, where $v_{i, j}$ has the following preference:
\begin{itemize}
\item For candidates not in the same row, her preference is
\[
c_{i, j_i} \succ c_{i + 1, j_{i+1}} \succ \cdots \succ c_{r, j_r} \succ c_{1, j_1} \succ \cdots \succ c_{i - 1, j_{i-1}}
\]
for any $j_1, j_2, \cdots, j_r$.
\item For candidates in the same row $i'$, her preference is
\[
c_{i', j} \succ c_{i', j + 1} \succ \cdots \succ c_{i', \ell} \succ c_{i', 1} \succ \cdots \succ c_{i', j - 1}.
\]
\end{itemize}
Let $K = r - 1$. For any deterministic committee $S_d$ of size $K$, there must be some $i \in [r]$, so that no candidate from the $i$-th row is in $S_d$, and at most one candidate from the $(i+1)$-th (or first if $i = r$) row is in $S_d$. Otherwise, every row where no candidate is selected must be followed by a row where at least $2$ candidate is selected. Thus, on average, at least $1$ candidate is selected from each row, contradicting with $K = r - 1$.

Let the candidate in the $(i+1)$-th row in $S_d$ be $c_{i+1,j}$ if there is one. Notice that $c_{i+1,j-1}$ is preferred to $S_d$ by at least $2\ell - 1$ voters: those in the $i$-th row and the $(i+1)$-th row except $v_{i+1,j}$. Therefore, the approximation ratio $c$ is at least $\frac{(2\ell - 1)K}{n} = \frac{(2\ell - 1)(r - 1)}{\ell r}$, which is close to $2$ when $\ell$ and $r$ are large.
\end{proof}
\section{Existence of Exactly Stable Lottery for $K \in \{1, 2, 3\}$}
\label{sec:three}
In this section, we strengthen the result in Section~\ref{sec:random} in the following special case: Each candidate has unit weight, and $K \in \{1, 2, 3\}$. There are $m$ candidates and $n$ voters with arbitrary monotone preferences over committees. In this setting, we show that there is a different way of constructing a dual solution that yields an exactly stable lottery. This opens up the possibility that the analysis in Section~\ref{sec:random} is not tight even for larger values of $K$. Indeed, we conjecture that there is an exactly stable lottery for any $K$ and any monotone preference structure.

\begin{reptheorem}{thm:small}
For unit-weight candidates and any number of voters with arbitrary monotone  preferences, when $K \in \{1, 2, 3\}$, an {\em exactly} stable lottery  always exists.
\end{reptheorem}

The $K = 1$ case is trivial. In the $K = 2$ case, w.l.o.g. we can assume the attacking strategy $\Delta_a$ only comprises committees of size $1$. This is because having a size-$K$ committee $S_a$ in $\Delta_a$ does not help the attacker even if all voters prefer $S_a$ to $S_d$. Then the $K = 2$ case is covered by Lemma~\ref{lem:same_size} below. Therefore we focus on the $K = 3$ case. We adopt the duality view introduced in Section~\ref{sec:randomized}. Given any attacking strategy $\Delta_a$, w.l.o.g. assume it only comprises committees of size $1$ and size $2$ for the same reason that having a size-$K$ committee in $\Delta_a$ does not help.

Let $p = \Pr_{S_a \sim \Delta_a} [|S_a| = 1]$, so that $\Pr_{S_a \sim \Delta_a} [|S_a| = 2] = 1 - p$. Note that the expected weight of $\Delta_a$ is $p + 2(1-p)$.

\paragraph{Case 1.} Suppose $p \in \{0,1\}$. In that case, all committees in the support of $\Delta_a$ have the same weight. The following lemma now shows the existence of an exactly stable lottery.

\begin{lemma}
\label{lem:same_size}
For any $K' \le K$, if every committee in the support of $\Delta_a$ has exactly the same weight $K'$, then an {\em exactly} stable lottery over committees of size at most $K$ always exists.
\end{lemma}
\begin{proof}
Given $\Delta_a$, we draw $S_1, S_2, \ldots, S_t$ independently from $\Delta_a$ where $t = \left\lfloor\frac{K}{K'}\right\rfloor$. Let $S_d = \bigcup_{i \in [t]} S_i$ so $|S_d| = tK' \leq K$. Now we need to show
\[
\E_{S_a \sim \Delta_a}\left[V(S_d, S_a) - \frac{K'}{K} \cdot n\right] < 0.
\]

For any voter $v \in [n]$, $\Pr_{S_a \sim \Delta_a}[S_a \succ_v S_d]$ is the probability that $S_a$ is strictly most preferred among $S_a, S_1, S_2, \ldots, S_t$, each of which is independently drawn from $\Delta_a$. Thus 
$$\Pr_{S_a \sim \Delta_a}[S_a \succ_v S_d] \le \frac{1}{t + 1} < \frac{K'}{K}.$$ 
Summing over $v \in \mathcal{N}$ gives the desired result.
\end{proof}

\paragraph{Case 2.} From now on, we will assume $p \in (0, 1)$. For convenience, we use $\Delta_1$ and $\Delta_2$ to denote the conditional distributions of $\Delta_a$ when $S_a$ is of weight $1$ and $2$, respectively. That is, for any $S \subseteq \mathcal{C}$,
\begin{align*}
\Pr_{S_1 \sim \Delta_1}[S_1 = S] = \Pr_{S_a \sim \Delta_a}[S_a = S \ | \ |S_a| = 1],\\
\Pr_{S_2 \sim \Delta_2}[S_2 = S] = \Pr_{S_a \sim \Delta_a}[S_a = S \ | \ |S_a| = 2].
\end{align*}
We construct a defending committee using the following procedure:
\begin{itemize}
\item With probability $p^2$, independently draw two committees from $\Delta_1$ and let $S_d$ be their union.
\item Otherwise (with probability $1 - p^2$), independently draw one committee from $\Delta_1$ and one from $\Delta_2$ and let $S_d$ be their union.
\end{itemize}
Denote the distribution of $S_d$ as $\Delta_d$, which is the defending strategy. To prove Theorem~\ref{thm:small}, we need to show
\begin{equation}
\E_{S_a \sim \Delta_a, S_d \sim \Delta_d}\left[V(S_d, S_a) \right] <  \frac{p + 2 (1 - p)}{3} \cdot n.\label{eq:thm_small_goal}
\end{equation}

Fix voter $v$. Consider the committees in decreasing order of voter preference. We say that a committee $S_a$ appears at position $x \in [0,1]$ if the total probability mass in $\Delta_a$ of committees $S \prec_v S_a$ is $x$. For convenience, we assume $x$ is continuous; this will only help the attacking strategy in the proof below. Similarly, let $f(x)$ denote the total probability mass in $\Delta_a$ of $S \prec_v S_a$ with $|S| = 1$, and $g(x)$ denote the total probability mass  in $\Delta_a$ of $S \prec_v S_a$ with $|S| = 2$. Clearly, $0 \leq f(x) \leq p$, $0 \leq g(x) \leq 1-p$, and $f(x) + g(x) \leq x$.

Now, if the attacker chooses $S_a$ at position $x$ (where $x$ is uniformly at random in $[0,1]$), the probability that the defender chooses one $S \prec_v S_a$ conditioned on $|S| = 1$ is $\frac{f(x)}{p}$. Similarly, the probability that the defender chooses one $S \prec_v S_a$ conditioned on $|S| = 2$ is $\frac{g(x)}{1-p}$. Since the defender chooses two committees of unit weight with probability $p^2$, and one committee of weight one and the other of weight $2$ with probability $1-p^2$, we have:


\begin{align}
\Pr_{S_a \sim \Delta_a, S_d \sim \Delta_d} [S_a \succ_v S_d] \leq & p^2 \int_0^1 \left(\frac{f(x)}{p}\right)^2 \mathrm{d}x + (1 - p^2) \int_0^1 \frac{f(x)}{p} \cdot \frac{g(x)}{1 - p} \mathrm{d}x\nonumber\\
\leq &\frac{1}{p} \cdot \int_0^1 \left((1 + p) \cdot x \cdot f(x) - f^2(x)\right) \mathrm{d}x,\label{eq:quadratic_goal}
\end{align}
where the first inequality enumerates the quantile of $S_a$ in $\Delta_a$.

To maximize (\ref{eq:quadratic_goal}), an integral over a quadratic function of $f(x)$ if we fix $p$ and $x$, we should have $f(x) = \min\left(p, \ \frac{(1 + p)x}{2}\right)$. Thus,
\begin{align*}
\Pr_{S_a \sim \Delta_a, S_d \sim \Delta_d} [S_a \succ_v S_d] \leq &\frac{1}{p} \cdot \left(\int_0^{\frac{2p}{1+p}} \frac{(1 + p)^2 \cdot x^2}{4} \mathrm{d}x + \int_{\frac{2p}{1+p}}^1 \left((1 + p) \cdot x \cdot p - p^2\right)\mathrm{d}x\right)\\
= &\frac{1}{p} \cdot \left(\frac{p^2}{3} + \frac{p(1 - p)}{2}\right)\\
= &\frac{1}{2} - \frac{p}{6} < \frac{2 - p}{3},
\end{align*}
where the final inequality follows since $p < 1$. Summing over the voters gives (\ref{eq:thm_small_goal}), and hence proves Theorem~\ref{thm:small}.

\end{document}